\documentclass[11pt]{article}
\usepackage[left=1in,top=1in,right=1in,bottom=1in,head=.1in,nofoot]{geometry}

\usepackage{setspace,url,bm,amsmath} 

\usepackage{titlesec} 
\titlelabel{\thetitle.\quad} 
\titleformat*{\section}{\bf\large\center\uppercase} 

\usepackage{graphicx} 
\usepackage{bbm}
\usepackage{latexsym}
\usepackage{caption}
\usepackage[margin=20pt]{subcaption}
\usepackage{hyperref}

\usepackage[table]{xcolor}

\newcommand{\GG}[1]{}

\usepackage{amsthm}
\theoremstyle{definition}

\newtheorem{prop}{Proposition}
\newtheorem{lemma}{Lemma}
\newtheorem{example}{Example}

\newtheorem*{corollary*}{Corollary}

\usepackage{natbib} 
\bibpunct{(}{)}{;}{a}{}{,} 

\usepackage{etoolbox} 
\apptocmd{\sloppy}{\hbadness 10000\relax}{}{} 

\usepackage{color}
\usepackage{listings}

\begin{document}
\doublespacing
\title{\bf On randomization-based causal inference for matched-pair factorial designs}
\author{Jiannan Lu\footnote{Address for correspondence: Jiannan Lu, One Microsoft Way, Redmond, Washington 98052-6399, U.S.A.
Email: \texttt{jiannl@microsoft.com}}~
and Alex Deng~\\Analysis and Experimentation, Microsoft Corporation}
\date{\today}
\maketitle
\begin{abstract}
Under the potential outcomes framework, we introduce matched-pair factorial designs, and propose the matched-pair estimator of the factorial effects. We also calculate the randomization-based covariance matrix of the matched-pair estimator, and provide the ``Neymanian'' estimator of the covariance matrix.
\end{abstract}

\section{Introduction}\label{sec:intro}

Randomization is widely regarded as the gold standard of causal inference \citep{Rubin:2008}. Under the potential outcomes framework \citep{Neyman:1923, Rubin:1974}, for a two-level factor, we define the causal effect as the linear contrast of the potential outcomes under treatment and control. To investigate multiple factors simultaneously, $2^K$ factorial designs \citep{Fisher:1935, Yates:1937} can be employed. Randomization-based casual inference for factorial designs has deep roots in the experimental design literature \cite[e.g.,][]{Kempthrone:1952}, and was recently presented using the language of potential outcomes \citep{Dasgupta:2015, Mukerjee:2016}.

Pair-matching \citep{Cochran:1953}, as a special form of stratification, has been widely adopted by researchers and practitioners \cite[e.g.,][]{Grossarth:2008}. For treatment-control studies (i.e., $2^1$ factorial designs), pair-matching has been extensively investigated by the causal inference community \citep{Rosenbaum:2002, Imai:2008, Imai:2009, Ding:2016, Fogarty:2016a, Fogarty:2016b}. Unfortunately, similar discussion appears to be missing for general factorial designs. In this paper, we fill this theoretical gap by extending \cite{Imai:2008}'s analysis to matched-pair factorial designs. We restrict the experimental units to be a fixed finite population, for a two-fold reason. First, as shown in \cite{Imai:2008}, it is straightforward to generalize the finite-population analyses to infinite populations. Second, for some practical examples, it might be unreasonable to view the experimental units as a random sample from an infinite population.

The paper proceeds as follows. Section \ref{sec:2k-c} reviews the randomization-based causal inference framework for completely randomized factorial designs. Section \ref{sec:2k-m} introduces matched-pair factorial designs, proposes the matched-pair estimator for the factorial effects, calculates its covariance matrix and the corresponding estimator. Section  \ref{sec:conclusion} briefly discusses the precision gains by pair-matching in factorial designs, and concludes.

\section{Causal inference for completely randomized factorial designs}\label{sec:2k-c}

To ensure self-containment, we first review the randomization-based causal inference framework for completely randomized factorial designs. Although most materials are adapted from \cite{Dasgupta:2015} and \cite{Lu:2016b, Lu:2016a}, some are refined for better clarity. For more detailed discussions on factorial designs, see, e.g., \cite{Wu:2009}.

\subsection{Factorial designs}

A $2^K$ factorial design consists of $K$ two-level (coded $-1$ and $+1$) factors. We represent it by the corresponding model matrix \citep{Wu:2009}, a $2^K \times 2^K$ matrix $\bm H_K = (\bm h_0, \ldots, \bm h_{2^K-1})$ that can  be constructed as follows: 
\begin{enumerate}

\item Let $\bm h_0 = \bm 1_{2^K};$ 

\item For $k=1,\ldots,K$, construct $\bm h_k$ by letting its first $2^{K-k}$ entries be $-1,$ the next $2^{K-k}$ entries be $+1,$ and repeating $2^{k-1}$ times;

\item If $K \ge 2,$ order all subsets of $\{1, \ldots, K\}$ with at least two elements, first by cardinality and then lexicography. For $k = 1, \ldots 2^K-K-1,$ let $\sigma_k$ be the $k$th subset and $\bm h_{K+k} = \prod_{l \in \sigma_k} \bm h_l,$ where ``$\prod$'' stands for entry-wise product.

\end{enumerate}

The use of the constructed $\bm H_K$ is two-fold:

\begin{enumerate}

\item $\bm h_0$ corresponds to the null effect; $\bm h_1$ to $\bm h_K$ correspond to the main effects of the $K$ factors; $\bm h_{K+1}$ to $\bm h_{K+\binom{K}{2}}$ correspond to the two-way interactions; $\ldots;$ $\bm h_{2^K-1}$ corresponds to the $K$-way interaction;

\item The $j$th row of $(\bm h_1, \ldots, \bm h_K)$ corresponds to the $j$th treatment combination $\bm z_j.$

\end{enumerate}

For $j=1, \dots, 2^K,$ let $\bm \lambda_j$ denote the $j$th row of $\bm H_K.$ 

\begin{example}
\label{example:1}
For $2^2$ factorial designs, the model matrix is:
\begin{equation*}
\bm H_2 =
\bordermatrix{& \bm h_0 & \bm h_1& \bm h_2 & \bm h_3\cr
             \bm \lambda_1 & +1  & -1 &  -1  & +1 \cr
             \bm \lambda_2 & +1 & -1 &  +1  & -1 \cr
             \bm \lambda_3 & +1  & +1 &  -1  & -1 \cr
             \bm \lambda_4 & +1 & +1 & +1 & +1}.
\end{equation*}
The four treatment combinations are $\bm z_1=(-1, -1),$ $\bm z_2=(-1, +1),$ $\bm z_3=(+1, -1)$ and $\bm z_4=(+1, +1).$ We represent the main effects of factors 1 and 2 by $\bm h_1 = (-1,-1,+1,+1)^\prime$ and $\bm h_2=(-1,+1,-1,+1)^\prime$ respectively, and the two-way interaction by $\bm h_3=(+1,-1,-1,+1)^\prime.$ 
\end{example}

\subsection{Randomization-based causal inference}

We consider a $2^K$ factorial design with $N = 2^K r$ units. By invoking the Stable Unit Treatment Value Assumption \citep{Rubin:1980}, for $i = 1, \ldots, N$ and $l=1, \ldots, 2^K,$ let the potential outcome of unit $i$ under $\bm z_l$ be $Y_i(\bm z_l),$ the average potential outcome for $\bm z_l$ be 
$
\bar Y(\bm z_l) = N^{-1} \sum_{i=1}^N Y_i(\bm z_l),
$
and
$
\bm Y_i = \{ Y_i(\bm z_1), \ldots, Y_i(\bm z_{2^K}) \}^{\prime}.
$
Define the individual and population-level factorial effect vectors as
\begin{equation}
\label{eq:estimand}
\bm \tau_i = \frac{1}{2^{K-1}} \bm H_K^\prime {\bm Y}_i
\quad
(i=1, \ldots, N);
\quad
\bm \tau = \frac{1}{N} \sum_{i=1}^N \bm \tau_i,
\end{equation}
respectively. Our interest lies in $\bm \tau.$

We denote the treatment assignment mechanism by 
\begin{equation*}
W_i(\bm z_l)
= 
\begin{cases}
1, & \text{if unit $i$ is assigned treatment $\bm z_l,$ } \\
0, & \text{otherwise.} \\
\end{cases}
\quad
(i=1, \ldots, N; l = 1, \ldots, 2^K).
\end{equation*}
We impose the following restrictions on the treatment assignment mechanism:
\begin{equation*}
\sum_{l=1}^{2^K} W_i(\bm z_l) = 1
\quad
(i = 1, \ldots, N);
\quad
\sum_{i=1}^N W_i(\bm z_l) = r
\quad
(l = 1, \ldots, 2^K).
\end{equation*}
In other words, we assign $r$ units to each treatment, and one treatment to each unit. Therefore, the observed outcome of unit $i$ is
$
Y_i^\textrm{obs} = \sum_{l=1}^{2^K} W_i(\bm z_l) Y_i(\bm z_l),
$
and the average observed outcome for treatment $\bm z_l$ is
$
\bar Y^\textrm{obs}(\bm z_l) 
=
r^{-1} \sum_{i=1}^N W_i(\bm z_l) Y_i(\bm z_l).
$
Under complete randomization, \cite{Dasgupta:2015} estimated $\bm \tau$ by
$$
\hat {\bm \tau}_{\textrm{C}} =  2^{-(K-1)} \bm H_K^\prime \bar{\bm Y}^{\mathrm{obs}},
\quad
\bar{\bm Y}^{\mathrm{obs}} = \{ \bar Y^\textrm{obs}(\bm z_1) , \ldots, \bar Y^\textrm{obs}(\bm z_{2^K}) \}^\prime.
$$
The sole source of randomness of 
$
\hat {\bm \tau}_{\textrm{C}}
$
is the treatment assignment. \cite{Dasgupta:2015} and \cite{Lu:2016a} derived the covariance matrix of this estimator, and the ``Neymanian'' estimator of the covariance matrix. We summarize their main results in the following lemmas.

\begin{lemma}
\label{lemma:cov-c}
$\hat {\bm \tau}_\textrm{C}$ is unbiased, and its covariance matrix is
\begin{equation}
\label{eq:cov-c}
\mathrm{Cov} ( \hat{\bm \tau}_\textrm{C} ) 
= \frac{1}{2^{2(K-1)}r} 
\sum_{l=1}^{2^K} \bm \lambda_l^\prime \bm \lambda_l \underbrace{\frac{1}{N-1} \sum_{i=1}^N \{ Y_i(\bm z_l) - \bar Y(\bm z_l) \}^2}_{S^2(\bm z_l)} 
- \frac{1}{N(N-1)} \sum_{i=1}^N (\bm \tau_i - \bm \tau)(\bm \tau_i - \bm \tau)^\prime.
\end{equation}
Moreover, the ``Neymanian'' estimator of the covariance matirx is
\begin{equation*}
\widehat \mathrm{Cov} ( \hat{\bm \tau}_\textrm{C} ) = \frac{1}{2^{2(K-1)}r} \sum_{l=1}^{2^K} \bm \lambda_l^\prime \bm \lambda_l \underbrace{\frac{1}{r-1}\sum_{i=1}^N W_i(\bm z_l) \{ Y_i^\textrm{obs} - \bar Y^\textrm{obs}(\bm z_l) \}^2}_{s^2(\bm z_l)},
\end{equation*}
whose bias is
$
\sum_{i=1}^N (\bm \tau_i - \bm \tau)(\bm \tau_i - \bm \tau)^\prime / (N^2-N).
$
\end{lemma}

The covariance matrix estimator 
$
\widehat \mathrm{Cov} ( \hat{\bm \tau}_\textrm{C} )
$
is ``conservative,'' because its diagonal entries, i.e., the variance estimators of the components of $\hat {\bm \tau}_\textrm{C},$ have non-negative biases.

\section{Causal inference for matched-pair randomized factorial designs}\label{sec:2k-m}

\subsection{Matched-pair designs and causal parameters}

As pointed out by \cite{Imai:2008}, they key idea behind matched-pair designs is that ``experimental units are paired based on their pre-treatment characteristics and the randomization of treatment is subsequently conducted within each matched pair.'' To apply this idea to factorial designs, we group the $N$ experimental units into $r$ ``pairs'' of $2^K$ units, and within each pair randomly assign one unit to each treatment. Let $\psi_j$ be the set of indices of the units in pair $j,$ such that
$$
|\psi_j| = 2^K
\quad
(j=1, \ldots, r);
\quad
\psi_j \cap \psi_{j^\prime} = \emptyset
\quad
(\forall j \neq j^\prime);
\quad
\cup_{j=1}^r \psi_j = \{1, \ldots, N\}.
$$
For pair $j,$ denote the average outcomes for treatment $\bm z_l$ as
$
\bar Y_{j\cdot}(\bm z_l) = 2^{-K} \sum_{i \in \psi_j} Y_i(\bm z_l),
$
and 
$
\bar{\bm Y}_{j\cdot} = \{ \bar Y_{j\cdot}(\bm z_1), \ldots, \bar Y_{j\cdot}(\bm z_{2^K}) \}^\prime,
$ 
and the factorial effect vector as
$
\bm \tau_{j\cdot} = 2^{-(K-1)} \bm H_K^\prime \bar{\bm Y}_{j\cdot}.
$
It is apparent
$$
\frac{1}{r} \sum_{j=1}^r \bar Y_{j\cdot}(\bm z_l) = \bar Y(\bm z_l)
\quad
(l = 1, \ldots, 2^k);
\quad
\frac{1}{r} \sum_{j=1}^r \bm \tau_{j\cdot} = \bm \tau.
$$

Within each pair, we randomly assign one unit to each treatment. Let the observed outcome of treatment $\bm z_l$ in pair $j$ be
$
Y_j^{\textrm{obs}}(\bm z_l) = \sum_{i \in \psi_j} Y_{i}(\bm z_l) W_i(\bm z_l),
$
and 
$
\bm Y_j^{\textrm{obs}} = \{ Y_j^{\textrm{obs}}(\bm z_1), \ldots, Y_j^{\textrm{obs}}(\bm z_{2^K}) \}^\prime.
$
We estimate $\bm \tau_{j\cdot}$ by 
$
\hat{\bm \tau}_{j\cdot} = 2^{-(K-1)} \bm H_K^\prime \bm Y_j^{\textrm{obs}}.
$
The matched-pair estimator for $\bm \tau$ is
\begin{equation}
\label{eq:est-m}
\hat {\bm \tau}_{\textrm{M}} = \frac{1}{r} \sum_{j=1}^r \hat{\bm \tau}_{j\cdot}.
\end{equation}

\subsection{Randomization-based inference}

We now present the main results of this paper.

\begin{prop}
\label{prop:cov-m}
$\hat {\bm \tau}_\textrm{M}$ is an unbiased estimator of $\bm \tau,$ and its covariance matrix is
\begin{equation}\label{eq:cov-m}
\mathrm{Cov} ( \hat{\bm \tau}_\textrm{M} ) 
= \frac{1}{2^{2(K-1)}r^2}
\sum_{l=1}^{2^K} \bm \lambda_l^\prime \bm \lambda_l \Delta_l 
- \frac{1}{2^K(2^K-1)r^2} \bm \Sigma,
\end{equation}
where
$$
\Delta_l = \frac{1}{2^K-1}
\left [
(N - 1) S^2(\bm z_l) - 2^K \sum_{j=1}^r 
\left\{ 
\bar Y_{j\cdot}(\bm z_l) - \bar Y(\bm z_l)
\right\}^2
\right ]
\quad
(l = 1, \ldots, 2^K),
$$
and
$$
\bm \Sigma
= \sum_{i=1}^N (\bm \tau_i - \bm \tau)(\bm \tau_i - \bm \tau)^\prime
- 2^K \sum_{j=1}^r (\bm \tau_{j\cdot} - \bm \tau)(\bm \tau_{j\cdot} - \bm \tau)^\prime.
$$
\end{prop}

\begin{proof}
To prove the first part, note that $\hat{\bm \tau}_{j\cdot}$ is an unbiased estimator of $\bm \tau_{j\cdot},$ for $j=1, \ldots, r.$ This fact combined with \eqref{eq:est-m} completes the proof. 

To prove the second part, let $\bm W_j = \{ W_i(\bm z_l) \}_{i\in \psi_j, l = 1, \ldots, 2^K}$ denote the treatment assignment for pair $j.$ By definition, $\bm W_j$'s are independently and identically distributed, implying the (joint) independence of 
$
\hat{\bm \tau}_{j\cdot}
$'s. 
Consequently, we can treat each pair as a completely randomized factorial design with $2^K$ units. Therefore by Lemma \ref{lemma:cov-c},
\begin{equation*}
\mathrm{Cov} ( \hat{\bm \tau}_{j\cdot} )
= 
\frac{1}{2^{2(K-1)}r^2}
\sum_{l=1}^{2^K} \bm \lambda_l^\prime \bm \lambda_l  \underbrace{\frac{1}{2^K-1} \sum_{i \in \psi_j} \{ Y_i(\bm z_l) - \bar Y_{j\cdot}(\bm z_l) \}^2}_{S_j^2(\bm z_l)} 
- \frac{1}{2^K(2^K-1)r^2} \sum_{i \in \psi_j } (\bm \tau_i - \bm \tau_{j\cdot})(\bm \tau_i - \bm \tau_{j\cdot})^\prime.
\end{equation*}
This implies that
\begin{eqnarray}
\label{eq:cov-m-proof}
\mathrm{Cov} ( \hat{\bm \tau}_\textrm{M} ) 
& = & \frac{1}{r^2} \sum_{j=1}^r \mathrm{Cov} ( \hat{\bm \tau}_{j\cdot} ) \nonumber \\
& = & 
\frac{1}{2^{2(K-1)}r^2}
\sum_{l=1}^{2^K} \bm \lambda_l^\prime \bm \lambda_l \sum_{j=1}^r S_j^2(\bm z_l) 
- \frac{1}{2^K(2^K-1)r^2} \sum_{j=1}^r \sum_{i \in \psi_j } (\bm \tau_i - \bm \tau_{j\cdot})(\bm \tau_i - \bm \tau_{j\cdot})^\prime.
\end{eqnarray}
To prove the equivalence between \eqref{eq:cov-m} and \eqref{eq:cov-m-proof}, simply note that
\begin{equation*}
(2^K - 1) \sum_{j=1}^r S_j^2(\bm z_l) 
+ 2^K \sum_{j=1}^r \{ \bar Y_{j\cdot}(\bm z_l) - \bar Y(\bm z_l)\}^2 
= (N-1) S^2(\bm z_l)
\end{equation*}
and
\begin{equation*}
\sum_{j=1}^r \sum_{i \in \psi_j } (\bm \tau_i - \bm \tau_{j\cdot})(\bm \tau_i - \bm \tau_{j\cdot})^\prime
+ 2^K \sum_{j=1}^r (\bm \tau_{j\cdot} - \bm \tau)(\bm \tau_{j\cdot} - \bm \tau)^\prime 
= \sum_{i=1}^N (\bm \tau_i - \bm \tau)(\bm \tau_i - \bm \tau)^\prime.
\end{equation*}
The proof is complete.
\end{proof}

We discuss a special case before moving forward. When $K=1,$ we have the classic treatment-control studies, and label the treatment and control as $+1$ and $-1,$ respectively. We are interested in the difference-in-mean estimator
\begin{equation*}
\hat {\tau}_{\textrm{MP}} = \frac{1}{r} \sum_{j=1}^r \{ Y_j^{\textrm{obs}}(+1) - Y_j^{\textrm{obs}}(-1) \}.
\end{equation*}
Denote $\psi_j = \{j_1, j_2\}.$ \cite{Imai:2008} (p. 4861, Eq. (8)) derived the variance of 
$
\hat {\tau}_{\textrm{MP}}
$
as
\begin{equation}
\label{eq:var-m-1}
\mathrm{Var} ( \hat {\tau}_{\textrm{MP}} ) = 
\frac{1}{4r^2} \sum_{j=1}^r 
\{ 
Y_{j_1}(+1) - Y_{j_2} (-1) - Y_{j_2} (+1) + Y_{j_1} (-1)
\}^2.
\end{equation}
As a validity check, Proposition \ref{prop:cov-m} reduces to \eqref{eq:var-m-1} when $K=1.$ We leave the proof to the readers.

We discuss the estimation of
$
\mathrm{Cov} ( \hat{\bm \tau}_\textrm{M} ), 
$
because Lemma \ref{lemma:cov-c} does not apply for matched-pair factorial designs. Inspired by \cite{Imai:2008}, we propose the following estimator:
\begin{equation}
\label{eq:cov-m-est}
\widehat{\mathrm{Cov}} ( \hat{\bm \tau}_\textrm{M} ) 
= \frac{1}{r(r-1)} 
\sum_{j=1}^r 
( \hat{\bm \tau}_{j\cdot} - \hat{\bm \tau}_\textrm{M} )( \hat{\bm \tau}_{j\cdot} - \hat{\bm \tau}_\textrm{M} )^\prime.
\end{equation}

\begin{prop}
\label{prop:cov-m-est-bias}
The bias of the covariance estimator in \eqref{eq:cov-m-est} is
\begin{equation*}
\mathrm{E} 
\left\{ 
\widehat{\mathrm{Cov}} ( \hat{\bm \tau}_\textrm{M} ) 
\right\} 
- \mathrm{Cov} ( \hat{\bm \tau}_\textrm{M} ) 
=
\frac{1}{r(r-1)}\sum_{j=1}^r (\bm \tau_{j\cdot} - \bm \tau) (\bm \tau_{j\cdot} - \bm \tau)^\prime.
\end{equation*}

\end{prop}

\begin{proof}
The proof is a basic maneuver of the expectation and covariance operators. First, by \eqref{eq:est-m} and the joint independence of
$
\hat{\bm \tau}_{j\cdot}
$'s,
$$
\mathrm{Cov} ( \hat{\bm \tau}_\textrm{M} ) 
=
r^{-2} \sum_{j=1}^r \mathrm{Cov} ( \hat{\bm \tau}_{j\cdot} ).
$$
Therefore by \eqref{eq:cov-m-est},
\begin{eqnarray*}
r(r-1)
\mathrm{E} 
\left\{ 
\widehat{\mathrm{Cov}} ( \hat{\bm \tau}_\textrm{M} ) 
\right\} 
& = &
\sum_{j=1}^r
\mathrm{E} 
( 
\hat{\bm \tau}_{j\cdot} \hat{\bm \tau}_{j\cdot}^\prime
)
- 
r\mathrm{E} 
( 
\hat{\bm \tau}_\textrm{M} \hat{\bm \tau}_\textrm{M}^\prime
)
\\
& = & 
\sum_{j = 1}^r \mathrm{Cov}(\hat{\bm \tau}_{j\cdot})
+
\sum_{j=1}^r \bm \tau_{j\cdot} \bm \tau_{j\cdot}^\prime
- r\mathrm{Cov} (\hat{\bm \tau}_\textrm{M}) 
-
r \bm \tau \bm \tau^\prime
\\
& = & 
r(r-1) \mathrm{Cov} (\hat{\bm \tau}_\textrm{M}) 
+
\sum_{j=1}^r 
(\bm \tau_{j\cdot} - \bm \tau) (\bm \tau_{j\cdot} - \bm \tau)^\prime.
\end{eqnarray*}
\end{proof}

Proposition \ref{prop:cov-m-est-bias} implies that the estimator of
$
\mathrm{Cov} ( \hat{\bm \tau}_\textrm{M} ) 
$
is also ``conservative.'' We leave it to the readers to prove that for treatment-control studies, Proposition \ref{prop:cov-m-est-bias} reduces to the corresponding results in \cite{Imai:2008} (p. 4862, Prop. 2, Part 1).

\section{Discussions and concluding remarks}\label{sec:conclusion}

For treatment-control studies, \cite{Imai:2008} compared the variance formulas for the complete-randomization and matched-pair estimators, and derived the condition under which pair-matching leads to precision gains. For general factorial designs, analogous comparisons can be made between \eqref{eq:cov-c} and \eqref{eq:cov-m}. However, to our best knowledge, intuitive closed-form expressions might not be available without additional assumptions on the potential outcomes.

There are multiple future directions based on our current work. First, we may compare the precisions of the complete-randomization and matched-pair estimators under certain mild restrictions on the potential outcomes. Second, it is possible to unify the randomization-based and regression-based inference frameworks, as pointed out by \cite{Samii:2012} and \cite{Lu:2016a}. Third, additional pre-treatment covariates may shed light on the pair-matching mechanism, and help sharpen our current analysis.

\section*{Acknowledgements}

The first author thanks Professor Tirthankar Dasgupta at Rutgers University and Professor Peng Ding at University California at Berkeley, for their early educations on causal inference and experimental design. We thank the Co-Editor-in-Chief and an anonymous reviewer for their thoughtful comments, which have substantially improved the presentation of this paper.

\bibliographystyle{apalike}
\bibliography{factorial_paired_match}

\begin{thebibliography}{}

\bibitem[Cochran, 1953]{Cochran:1953}
Cochran, W.~G. (1953).
\newblock Matching in analytical studies.
\newblock {\em American Journal of Public Health}, 43:684--691.

\bibitem[Dasgupta et~al., 2015]{Dasgupta:2015}
Dasgupta, T., Pillai, N., and Rubin, D.~B. (2015).
\newblock Causal inference from $2^k$ factorial designs using the potential
  outcomes model.
\newblock {\em Journal of the Royal Statistical Society: Series B},
  77:727--753.

\bibitem[Ding, 2016]{Ding:2016}
Ding, P. (2016).
\newblock A paradox from randomization-based causal inference (with
  discussion).
\newblock {\em Statistical Science}, in press.

\bibitem[Fisher, 1935]{Fisher:1935}
Fisher, R.~A. (1935).
\newblock {\em The Design of Experiments}.
\newblock Edinburgh: Oliver and Boyd.

\bibitem[Fogarty, 2016a]{Fogarty:2016a}
Fogarty, C.~B. (2016a).
\newblock Regression assisted inference for the average treatment effect in
  paired experiments.
\newblock {\em arXiv:1612.05179}.

\bibitem[Fogarty, 2016b]{Fogarty:2016b}
Fogarty, C.~B. (2016b).
\newblock Sensitivity analysis for the average treatment effect in paired
  observational studies.
\newblock {\em arXiv:1609.02112}.

\bibitem[Grossarth-Maticek and Ziegler, 2008]{Grossarth:2008}
Grossarth-Maticek, R. and Ziegler, R. (2008).
\newblock Randomized and non-randomized prospective controlled cohort studies
  in matched pair design for the long-term therapy of corpus uteri cancer
  patients with a mistletoe preparation.
\newblock {\em European Journal of Medical Research}, 13:107--120.

\bibitem[Imai, 2008]{Imai:2008}
Imai, K. (2008).
\newblock Variance identification and efficiency analysis in randomized
  experiments under the matched-pair design.
\newblock {\em Statistics in Medicine}, 27:4857--4873.

\bibitem[Imai et~al., 2009]{Imai:2009}
Imai, K., King, G., and Nall, C. (2009).
\newblock The essential role of pair matching in cluster-randomized
  experiments, with application to the mexican universal health insurance
  evaluation (with discussion).
\newblock {\em Statistical Science}, 24:29--53.

\bibitem[Kempthrone, 1952]{Kempthrone:1952}
Kempthrone, O. (1952).
\newblock {\em The Design and Analysis of Experiments}.
\newblock New York: Wiley.

\bibitem[Lu, 2016a]{Lu:2016b}
Lu, J. (2016a).
\newblock Covariate adjustment in randomization-based causal inference for
  $2^k$ factorial designs.
\newblock {\em Statistics \& Probability Letters}, 119:11--20.

\bibitem[Lu, 2016b]{Lu:2016a}
Lu, J. (2016b).
\newblock On randomization-based and regression-based inferences for $2^k$
  factorial designs.
\newblock {\em Statistics \& Probability Letters}, 112:72--78.

\bibitem[Mukerjee et~al., 2016]{Mukerjee:2016}
Mukerjee, R., Dasgupta, T., and Rubin, D.~B. (2016).
\newblock Causal inference in rebuilding and extending the recondite bridge
  between finite population sampling and experimental design.
\newblock {\em arXiv:1606.05279}.

\bibitem[Neyman, 1923]{Neyman:1923}
Neyman, J.~S. (1990[1923]).
\newblock On the application of probability theory to agricultural experiments.
  essay on principles (with discussion). section 9 (translated). reprinted ed.
\newblock {\em Statistical Science}, 5:465--472.

\bibitem[Rosenbaum, 2002]{Rosenbaum:2002}
Rosenbaum, P.~R. (2002).
\newblock {\em Observational Studies, 2nd Edition}.
\newblock Springer.

\bibitem[Rubin, 1974]{Rubin:1974}
Rubin, D.~B. (1974).
\newblock Estimating causal effects of treatments in randomized and
  nonrandomized studies.
\newblock {\em Journal of Educational Psychology}, 66:688--701.

\bibitem[Rubin, 1980]{Rubin:1980}
Rubin, D.~B. (1980).
\newblock Comment on ``{R}andomized analysis of experimental data: The {F}isher
  randomization test'' by {D}. {B}asu.
\newblock {\em Journal of American Statistical Association}, 75:591--593.

\bibitem[Rubin, 2008]{Rubin:2008}
Rubin, D.~B. (2008).
\newblock For objective causal inference, design trumps analysis.
\newblock {\em The Annals of Applied Statistics}, pages 808--840.

\bibitem[Samii and Aronow, 2012]{Samii:2012}
Samii, C. and Aronow, P.~M. (2012).
\newblock On equivalencies between design-based and regression-based variance
  estimators for randomized experiments.
\newblock {\em Statistics and Probability Letters}, 82:365--370.

\bibitem[Wu and Hamada, 2009]{Wu:2009}
Wu, C. F.~J. and Hamada, M.~S. (2009).
\newblock {\em Experiments: Planning, Analysis, and Optimization}.
\newblock New York: Wiley.

\bibitem[Yates, 1937]{Yates:1937}
Yates, F. (1937).
\newblock The design and analysis of factorial experiments.
\newblock {\em Technical Communication}, 35.
\newblock Imperial Bureau of Soil Science, London.

\end{thebibliography}

\end{document}